\documentclass[12pt]{article}
\usepackage{blindtext}
\usepackage{graphicx}
\usepackage{subfigure}
\usepackage{dsfont}
\usepackage{setspace,graphicx,epstopdf,amsmath,amsfonts,amssymb,amsthm,versionPO}
\usepackage{marginnote,datetime,enumitem,subfigure,rotating,fancyvrb}
\usepackage{float}
\usepackage[longnamesfirst]{natbib}
\usdate

\excludeversion{notes}		
\includeversion{links}          

\iflinks{}{\hypersetup{draft=true}}

\ifnotes{%
\usepackage[margin=1in,paperwidth=10in,right=2.5in]{geometry}%
\usepackage[textwidth=1.4in,shadow,colorinlistoftodos]{todonotes}%
}{%
\usepackage[margin=1in]{geometry}%
\usepackage[disable]{todonotes}%
}

\makeatletter\let\chapter\@undefined\makeatother 

\usepackage{indentfirst} 
\usepackage{jfe}          

\newtheorem{proposition}{Proposition}

\newtheorem{lemma}{Lemma}[section]

\begin{document}

\setlist{noitemsep}  

\title{Closed-form Solutions of Relativistic Black-Scholes Equations\footnotetext{}}

\author{Yanlin Qu$^{1,2}$ and Randall R. Rojas$^{2}$ \\ \\$^{1}$School of Mathematical Science\\ University of Science and Technology of China\\ \\ $^{2}$Department of Economics\\University of
California, Los Angeles \\ {\small \tt quyanlin$@$mail.ustc.edu.cn, rrojas$@$econ.ucla.edu}}

\date{}              


\renewcommand{\thefootnote}{\fnsymbol{footnote}}

\singlespacing

\maketitle

\vspace{-.2in}
\begin{abstract}
\noindent Drawing insights from the triumph of relativistic over classical mechanics when velocities approach the speed of light, we explore a similar improvement to the seminal Black-Scholes (\cite{black1973pricing}) option pricing formula by considering a relativist version of it, and then finding a respective solution. We show that our solution offers a significant improvement over competing solutions (e.g., \cite{2016arXiv160401447R}), and obtain a new closed-form option pricing formula, containing the speed limit of information transfer $c$ as a new parameter. The new formula is rigorously shown to converge to the Black-Scholes formula as $c$ goes to infinity. When $c$ is finite, the new formula can flatten the standard volatility smile which is more consistent with empirical observations. In addition, an alternative family of distributions for stock prices arises from our new formula, which offer a better fit, are shown to converge to lognormal, and help to better explain the volatility skew.
\end{abstract}

\medskip


\thispagestyle{empty}

\clearpage

\onehalfspacing
\setcounter{footnote}{0}
\renewcommand{\thefootnote}{\arabic{footnote}}
\setcounter{page}{1}

\section{Introduction}

In agreement with the work done by \cite{2016arXiv160401447R}, we take a similar approach for the initial formulation. However, for the purpose of clarity, we simplify the notations by defining
$$\alpha=\frac{1}{\sigma^2}\left(\frac{\sigma^2}{2}-r\right),\;\beta=\frac{1}{2\sigma^2}\left(\frac{\sigma^2}{2}+r\right)^2,\;\nu=\frac{c^2}{\sigma^2}\;,$$ where $r$ is the risk-free interest rate, $\sigma$ is the volatility in the Black-Scholes model, and $c$ is the speed of light. As usual, let $S$, $K$, $T$ be the current stock price, strike price and maturity in the Black-Scholes model respectively. In addition, $m$ represents a mass and $\hbar$ is the reduced Planck constant. It is well-known that the Black-Scholes equation
\begin{equation}
\frac{\partial f}{\partial t}+rS\frac{\partial f}{\partial S}+\frac{1}{2}\sigma^2S^2\frac{\partial^2f}{\partial S^2}=rf
\end{equation}
can be mapped to the free Schr\"{o}dinger equation
\begin{equation}
i\hbar\frac{\partial\psi}{\partial\tilde{t}}=-\frac{\hbar^2}{2m}\frac{\partial^2\psi}{\partial x^2}
\end{equation}
by
$$\tilde{t}=it,\;\hbar=1,\;m=\frac{1}{\sigma^2},\;x=\log S,\;\psi=e^{-(\alpha x+\beta t)}\cdot f.$$
\noindent In relativistic quantum mechanics, the free Schr\"{o}dinger equation is replaced by the Klein-Gordon equation
\begin{equation}
-\frac{\hbar^2}{c^2}\frac{\partial^2\tilde{\psi}}{\partial\tilde{t}^2}+\hbar^2\frac{\partial^2\tilde{\psi}}{\partial x^2}=m^2c^2\tilde{\psi}.
\end{equation}
in the sense that $\tilde{\psi}\cdot e^{imc^2\tilde{t}/\hbar}\rightarrow\psi$ as $c\rightarrow\infty$ (see \cite{10.2307/24890097}). Therefore we apply the inverse map between (1) and (2) with
$$\tilde{\psi}=e^{-imc^2\tilde{t}/\hbar}\cdot\psi=e^{\nu t}\cdot\psi=e^{-(\alpha x+(\beta-\nu) t)}\cdot f$$
on (3) to obtain the relativistic generalized Black-Scholes equation
\begin{equation}
\frac{1}{2\nu}\frac{\partial^2f}{\partial t^2}+\left(1-\frac{\beta}{\nu}\right)\frac{\partial f}{\partial t}+rS\frac{\partial f}{\partial S}+\frac{1}{2}\sigma^2S^2\frac{\partial^2f}{\partial S^2}=\left(r-\frac{\beta^2}{2\nu}\right)f.
\end{equation}
It is clear that (4)$\rightarrow$(1) as $\nu\rightarrow\infty$. In that paper, the authors then tried to solve (4) approximately by ignoring several important terms. However, even with such simplifications, the integral in their solution fails to converge for the plain vanilla options cases.

In Section 2 we find an explicit solution of (4) which leads to a closed-form option pricing formula, and provide a rigorous proof of convergence as $c\rightarrow\infty$. In Section 3, we show how the new formula and corresponding distribution can flatten the volatility smile/skew from both a theoretical foundation and empirical tests.

\section{The Model} \label{sec:Model}

\subsection{Solving the Equation}

To illustrate our technique to solve (4) explicitly, we only consider the call option case with $\alpha>0$. Other cases can be solved in a similar way. By adding the terminal condition to (3) and replacing $\tilde{t}$ by $it$, we instead need to solve
\begin{equation}
\left\{
\begin{aligned}
&\frac{\partial^2\tilde{\psi}}{\partial t^2}+c^2\frac{\partial^2\tilde{\psi}}{\partial x^2}=\nu^2\tilde{\psi},\;\;\;\;\;\;(x,t)\in\mathbb{R}\times(0,T)\\
&\tilde{\psi}_T=e^{-(\alpha x+(\beta-\nu)T)}(e^x-K)^+\\
\end{aligned}
\right..
\end{equation}
Let $\varphi(x,t)=\tilde{\psi}(cx,T-t)$, then the terminal condition becomes the initial condition and we can solve the equation in the upper half plane $\mathbb{H}=\mathbb{R}\times\mathbb{R}^+$ instead of $\mathbb{R}\times(0,T)$. (5) now can be written as
\begin{equation}
\left\{
\begin{aligned}
&\Delta\varphi=\nu^2\varphi,\;\;\;\;\;\;(x,t)\in \mathbb{H}\\
&\varphi_0=e^{-(\beta-\nu)T}(e^{(1-\alpha)cx}-Ke^{-\alpha cx})\mathds{1}_{cx>\log K}\\
\end{aligned}
\right.
\end{equation}
where $\Delta$ is the Laplacian and $\mathds{1}$ is the indicator function. Since $\alpha<\frac{1}{2}$ by definition, $\varphi_0$  grows exponentially as $x\rightarrow +\infty$ so that we can not apply the Fourier Transform to it. Also, it is not smooth at $\frac{\log K}{c}$ which makes it impossible to obtain an intuitive solution. Therefore, the main insight of our technique is to separate $\varphi_0$ into two parts, a simple part $\varphi_{10}$ and an integrable part $\varphi_{20}$ with $\varphi_0=\varphi_{10}-\varphi_{20}$ and $\varphi=\varphi_{1}-\varphi_{2}$, where both can be solved analytically. Figure 1 illustrates our proposed decomposition.
\begin{figure}[H]
  \centering
  \includegraphics[width=0.5\textwidth]{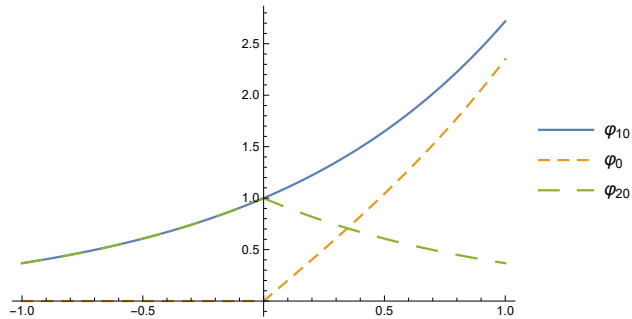}
  \caption{An example of the decomposition}
\end{figure}
For the simple part,
\begin{equation}
\left\{
\begin{aligned}
&\Delta\varphi_{1}=\nu^2\varphi_{1},\;\;\;\;\;\;(x,t)\in \mathbb{H}\\
&\varphi_{10}=e^{-(\beta-\nu)T+(1-\alpha)cx}\\
\end{aligned}
\right.
\end{equation}
we can derive a pair of solutions after several simple observations, namely,
$$\varphi_1=e^{-(\beta-\nu)T+(1-\alpha)cx\pm\sqrt{\nu^2-(1-\alpha)^2c^2}\cdot t}$$
where we need $c>c_0=\sigma^2/2+r$ to guarantee the square root is real. We now substitute $\varphi_1$ for $f_1$ in (4),
$$f_{10}=\tilde{\psi}_1(x,0)e^{\alpha x}=\varphi_1(x/c,T)e^{\alpha x}=Se^{(\pm\sqrt{\nu^2-(1-\alpha)^2c^2}-(\beta-\nu))T}.$$
If we choose $``+"$, $f_{10}\rightarrow\infty$ as $c\rightarrow\infty$. Since $f_{20}$ corresponds to the integrable part which is bounded, it is impossible for $f_0=f_{10}-f_{20}$ to converge to the Black-Scholes formula as $c\rightarrow\infty$. Therefore $``-"$ is the only choice,
$$f_{10}=Se^{(-\sqrt{\nu^2-(1-\alpha)^2c^2}-(\beta-\nu))T}.$$

For the integrable part,
\begin{equation}
\left\{
\begin{aligned}
&\Delta\varphi_{2}=\nu^2\varphi_{2},\;\;\;\;\;\;(x,t)\in \mathbb{H}\\
&\varphi_{20}=e^{-(\beta-\nu)T}(e^{(1-\alpha)cx}\mathds{1}_{cx\leq\log K}+Ke^{-\alpha cx}\mathds{1}_{cx>\log K})\\
\end{aligned}
\right.
\end{equation}
notice that $0<\alpha<1/2$ so that $\varphi_{20}$ has an exponential decay as $x\rightarrow\pm\infty$. Therefore we can apply a Fourier Transform with respect to $x$ to obtain
\begin{equation}
\left\{
\begin{aligned}
&\frac{\partial^2\hat{\varphi}_2}{\partial t^2}(\xi,t)=(\nu^2+\xi^2)\hat{\varphi}_2(\xi,t),\;\;\;\;\;\;(\xi,t)\in\mathbb{H}\\
&\hat{\varphi}_{20}(\xi)=\frac{e^{-(\beta-\nu)T}K^{(1-\alpha)-\frac{i\xi}{c}}}{c\sqrt{2\pi}((1-\alpha)-\frac{i\xi}{c})(\alpha+\frac{i\xi}{c}))}\\
\end{aligned}
\right.
\end{equation}
where we have used the following two simple formulas $(a>0)$

\begin{equation}
\begin{aligned}
&\widehat{e^{-at}\mathds{1}_{t>k}}=\frac{1}{\sqrt{2\pi}}\int_\mathbb{R}e^{-at}\mathds{1}_{t>k}\cdot e^{-i\omega t}dt
=\frac{1}{\sqrt{2\pi}}\int_k^\infty e^{-(a+i\omega)t}dt=\frac{e^{-k(a+i\omega)}}{\sqrt{2\pi}(a+i\omega)},\\
\\
&\widehat{e^{at}\mathds{1}_{t\leq k}}=\frac{1}{\sqrt{2\pi}}\int_\mathbb{R}e^{at}\mathds{1}_{t\leq k}\cdot e^{-i\omega t}dt
=\frac{1}{\sqrt{2\pi}}\int_{-\infty}^k e^{(a-i\omega)t}dt=\frac{e^{k(a-i\omega)}}{\sqrt{2\pi}(a-i\omega)}.\\
\end{aligned}\nonumber
\end{equation}

Solutions of (9) should be $\hat{\varphi}_2=e^{\pm\sqrt{\nu^2+\xi^2}\cdot t}\hat{\varphi}_{20}.$
In order to apply the inverse Fourier Transform on $\hat{\varphi}_2$ to obtain $\varphi_2$, the solutions of (8), $\hat{\varphi}_2$ must be integrable with respect to $\xi$, so $``-"$ is the only choice
$$\varphi_2(x,t)=\frac{1}{\sqrt{2\pi}}\int_\mathbb{R}\hat{\varphi}_{20}(\xi)\cdot e^{-\sqrt{\nu^2+\xi^2}\cdot t+ix\xi}d\xi.$$
We now return $\varphi_2$ to $f_2$ in (4) with $y=\xi/c$
$$f_{20}=\tilde{\psi}_2(x,0)e^{\alpha x}=\varphi_2(x/c,T)e^{\alpha x}=\frac{e^{-\beta T}}{2\pi}\frac{S^\alpha}{K^{\alpha-1}}\int_\mathbb{R}\frac{(S/K)^{iy}\cdot e^{-(\sqrt{c^2y^2+\nu^2}-\nu)T}}{(\alpha+iy)(1-\alpha-iy)}dy.$$
Finally, we get the formula for this case,
$$f_0=f_{10}-f_{20}=e^{-\beta T}\left(Se^{(\nu-\sqrt{\nu^2-c^2(1-\alpha)^2})T}+\frac{1}{2\pi}\frac{S^\alpha}{K^{\alpha-1}}\int_\mathbb{R}\frac{(S/K)^{iy}e^{-(\sqrt{c^2y^2+\nu^2}-\nu)T}}{(\alpha+iy)(\alpha+iy-1)}dy\right).$$
\subsection{New Formulas}
For $\alpha\neq 0$, our technique is valid for both European call and put options. Therefore, after repeating the procedure above, we have
\begin{equation}
\begin{aligned}
&C_c(r,\sigma,S,T,K)=e^{-\beta T}\Biggl(Se^{(\nu-\sqrt{\nu^2-c^2(1-\alpha)^2})T}-Ke^{(\nu-\sqrt{\nu^2-c^2\alpha^2})T}\mathds{1}_{\alpha<0}\\
&\;\;\;\;\;\;\;\;\;\;\;\;\;\;\;\;\;\;\;\;\;\;\;\;\;\;\;\;\;\;+\frac{1}{2\pi}\frac{S^\alpha}{K^{\alpha-1}}\int_\mathbb{R}\frac{(S/K)^{iy}e^{-(\sqrt{c^2y^2+\nu^2}-\nu)T}}{(\alpha+iy)(\alpha+iy-1)}dy\Biggl)\\
&P_c(r,\sigma,S,T,K)=e^{-\beta T} \Biggl(Ke^{(\nu-\sqrt{\nu^2-c^2\alpha^2})T}\mathds{1}_{\alpha>0}\\
&\;\;\;\;\;\;\;\;\;\;\;\;\;\;\;\;\;\;\;\;\;\;\;\;\;\;\;\;\;\;+\frac{1}{2\pi}\frac{S^\alpha}{K^{\alpha-1}}\int_\mathbb{R}\frac{(S/K)^{iy}e^{-(\sqrt{c^2y^2+\nu^2}-\nu)T}}{(\alpha+iy)(\alpha+iy-1)}dy \Biggl).
\end{aligned}
\end{equation}

However, our technique fails when $\alpha=0$. In practice, with a proper computational accuracy, it is almost impossible for $\alpha=0$ to hold exactly. Even if it does hold, since we never know the exact value of $\sigma$, we can add a small perturbation to $\sigma$  to avoid this case. Although there is an improper integral in each formula, notice that the integrand has an exponential decay so that we can use numerical integration on finite intervals to approximate it with high efficiency.

In addition, we have a new put-call parity,
$$C_c(r,\sigma,S,T,K)-P_c(r,\sigma,S,T,K)=e^{-\beta T}(Se^{(\nu-\sqrt{\nu^2-c^2(1-\alpha)^2})T}-Ke^{(\nu-\sqrt{\nu^2-c^2\alpha^2})T}).$$

\subsection{Convergence}
Since $(4)\rightarrow(1)$ as $c\rightarrow\infty$, our new formulas should converge to the Black-Scholes formulas. However, in Section 2.1, we had ruled out several solutions of (4) because they can not converge to the solutions of (1). Therefore, it is possible that our new formulas fail to converge as well. Fortunately, we can eliminate such possibility in this section by providing a rigorous mathematical proof of the convergence. To illustrate the method of the proof explicitly, we only consider the put option case with $\alpha<0$. Other cases can be proved in a similar way. Before we start, it is convenient for us to recall the Black-Scholes formulas,
\begin{equation}
C(r,\sigma,S,T,K)=SN(d_1)-Ke^{-rT}N(d_2)
\end{equation}
\begin{equation}
P(r,\sigma,S,T,K)=Ke^{-rT}N(-d_2)-SN(-d_1)
\end{equation}
where
\begin{equation}
d_1=\frac{\log{S/K}+(r+\frac{\sigma^2}{2})T}{\sigma\sqrt{T}}\;\;\;
d_2=\frac{\log{S/K}+(r-\frac{\sigma^2}{2})T}{\sigma\sqrt{T}}.
\end{equation}
To build a connection between $N(\cdot)$ (the CDF of the normal distribution) in (11) and (12) and the improper integral in (10), we establish the following interesting lemma.
\begin{lemma}
Let $\theta>0$ then
\begin{equation}
N(\tau)=\frac{i}{2\pi}\int_\mathbb{R}\frac{e^{-\frac{1}{2}(x+i\theta)^2-i\tau(x+i\theta)}}{x+i\theta}dx,\;\;\forall\tau\in\mathbb{R}.
\end{equation}
\end{lemma}
\begin{proof}
Denote the integrand above with $f(x,\tau)$ (including the constant before the integral) so
$$\frac{\partial f}{\partial \tau}(x,\tau)=\frac{1}{2\pi}e^{-\frac{1}{2}(x+i\theta)^2-i\tau(x+i\theta)}.$$
Both $f$ and $\frac{\partial f}{\partial \tau}$ are continuous in $\mathbb{R}$. In addition, $\int_\mathbb{R}\frac{\partial f}{\partial \tau}(x,\tau)dx$ converges uniformly in any finite interval of $\tau$ since
$$\left|\frac{\partial f}{\partial \tau}(x,\tau)\right|\leq e^{-\frac{1}{2}x^2+\frac{1}{2}\theta^2+\tau\theta}.$$
By Theorem 11 in \cite{trench2012functions}, we have
$$\frac{\partial}{\partial\tau}\int_\mathbb{R}f(x,\tau)dx=\int_\mathbb{R}\frac{\partial f}{\partial\tau}(x,\tau)dx=e^{\tau\theta+\frac{\theta^2}{2}}\frac{1}{2\pi}\int_{\mathbb{R}}e^{-\frac{x^2}{2}-ix(\tau+\theta)}dx=\frac{1}{\sqrt{2\pi}}e^{-\frac{\tau^2}{2}}$$
where the last equality holds because of the Fourier Transform from $x$ to $\tau+\theta$. Therefore the derivatives of both sides of (14) are always the same so we only need to prove that (14) holds when $\tau=0$. Consider the complex integral of $$h(z)=\frac{i}{2\pi}\frac{e^{-\frac{z^2}{2}}}{z}$$ on the rectangular contour $\Gamma_R$ in the following figure (anticlockwise).
\begin{figure}[H]
  \centering
  \includegraphics[width=0.5\textwidth]{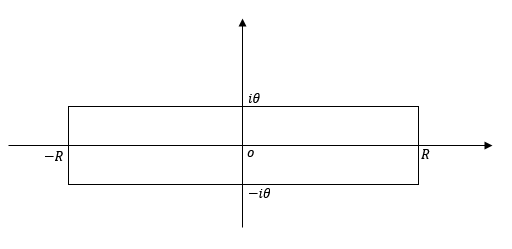}
  \caption{The rectangular contour $\Gamma_R$}
\end{figure}
\noindent  Notice that $h$ has a simple pole at the origin, by residue theorem,
$$\oint_{\Gamma_R}h(z)dz=2\pi iRes(h,0)=2\pi i\lim_{z\rightarrow 0}zh(z)=-1.$$
Since $$|h(z)|=\left|\frac{i}{2\pi}\frac{e^{-\frac{(x+iy)^2}{2}}}{x+iy}\right|\leq\frac{1}{2\pi}\frac{e^{-\frac{1}{2}R^2+\frac{1}{2}\theta^2}}{R}\;{\rm{for}}\;x=\pm R\;and\;|y|\leq\theta,$$
the integral on two vertical segments vanishes as $R\rightarrow\infty$.
Notice that the integral on two horizontal segments are the same because $h(z)=-h(-z)$. Let $R\rightarrow\infty$ and finally we have
$$\int_\mathbb{R}f(x,0)dx=-\frac{1}{2}\lim_{R\rightarrow\infty}\oint_{\Gamma_R}h(z)dz=\frac{1}{2}=N(0).$$
\end{proof}
With the lemma proved, we can prove the convergence now.
\begin{proposition}
Let $\alpha<0$ then $P_c\rightarrow P$ as $c\rightarrow\infty$.
\end{proposition}
\begin{proof}
Denote the integrand in (10) with $f(y,c)$ (including the constant before the integral). We have
\begin{equation}
\lim_{c\rightarrow\infty}\sqrt{c^2y^2+\nu^2}-\nu=\lim_{c\rightarrow\infty}\frac{c^2y^2}{\sqrt{c^2y^2+\nu^2}+\nu}
=\lim_{c\rightarrow\infty}\frac{y^2\sigma^2}{\sqrt{\frac{y^2\sigma^4}{c^2}+1}+1}=\frac{y^2\sigma^2}{2}
\end{equation}
and then
$$|f(y,c)|\leq\frac{e^{-\beta T}}{2\pi}\frac{S^\alpha}{K^{\alpha-1}}\frac{e^{-(\sqrt{c^2y^2+\nu^2}-\nu)T}}{|\alpha||1-\alpha|}\leq\widetilde{C}
e^{-\frac{y^2\sigma^2T}{\sqrt{\frac{y^2\sigma^4}{c_0^2}+1}+1}}$$
where $\widetilde{C}$ is a constant and $c_0=\sigma^2/2+r$ is the lower bound of $c$. Therefore, $f$ is dominated by a $c$-independent function with an exponential decay so that $\int_\mathbb{R}f(y,c)dy$ converges uniformly in $[c_0,\infty)$.
By theorem 10 in \cite{trench2012functions} and (15),
\begin{equation}
\begin{aligned}
\lim_{c\rightarrow\infty}P_c&=\lim_{c\rightarrow\infty}\int_\mathbb{R}f(y,c)dy=\int_\mathbb{R}\lim_{c\rightarrow\infty}f(y,c)dy\\
&=\frac{e^{-\beta T}}{2\pi}\frac{S^\alpha}{K^{\alpha-1}}\int_\mathbb{R}\frac{(S/K)^{iy}e^{-\frac{y^2\sigma^2T}{2}}}{(\alpha+iy)(\alpha+iy-1)}dy\\
&=-\frac{e^{-\beta T}}{2\pi}\frac{S^\alpha}{K^{\alpha-1}}\int_\mathbb{R}\frac{(S/K)^{iy}e^{-\frac{y^2\sigma^2T}{2}}}{\alpha+iy}dy
+\frac{e^{-\beta T}}{2\pi}\frac{S^\alpha}{K^{\alpha-1}}\int_\mathbb{R}\frac{(S/K)^{iy}e^{-\frac{y^2\sigma^2T}{2}}}{\alpha+iy-1}dy
\end{aligned}
\end{equation}
Let $\theta=-\alpha\sqrt{\sigma^2T}$, $\tau=-d_2$, $x=y\sqrt{\sigma^2T}$ in lemma 2.1, after standard simplification, we will find that the first term in (16) equals to the first term in (12). Let $\theta=(1-\alpha)\sqrt{\sigma^2T}$, $\tau=-d_1$, $x=y\sqrt{\sigma^2T}$ in lemma 2.1, after standard simplification, we will find that the second term in (16) equals to the second term in (12). Then the proposition is proved.
\end{proof}
\section{Distributions And Volatility Smile/Skew}
\subsection{Flatten Standard Smiles Directly}
After proving the convergence, now we can turn to study the difference between our new formulas and the Black-Scholes formulas. We mainly study formulas for call options here. If we fix $r=0.1,\;\sigma=0.5,\;S=50,\;T=1,\;c=3$, then $C-C_c$ is a function of $K$ which is plotted in the following Figure 3(a). Figure 3(b) shows that $C_c$ is a monotonic increasing function of $\sigma$ like $C$. If we meet a standard volatility smile, the implied volatility will be lower than $\sigma$ in $C$ when $K\approx S$ and will be higher when $|K-S|\gg 0$. By the monotonicity, it indicates that the Black-Scholes formula overprices options when $K\approx S$ and underprices options when $|K-S|\gg 0$. From Figure 3(a) we can see that our new formula can reduce this mispricing almost perfectly. We can expect that our new formula can flatten the smile.
\begin{figure}[htbp]
\centering
  \subfigure[The difference between two formulas]{
  \begin{minipage}{7cm}
  \centering
  \includegraphics[width=1\textwidth]{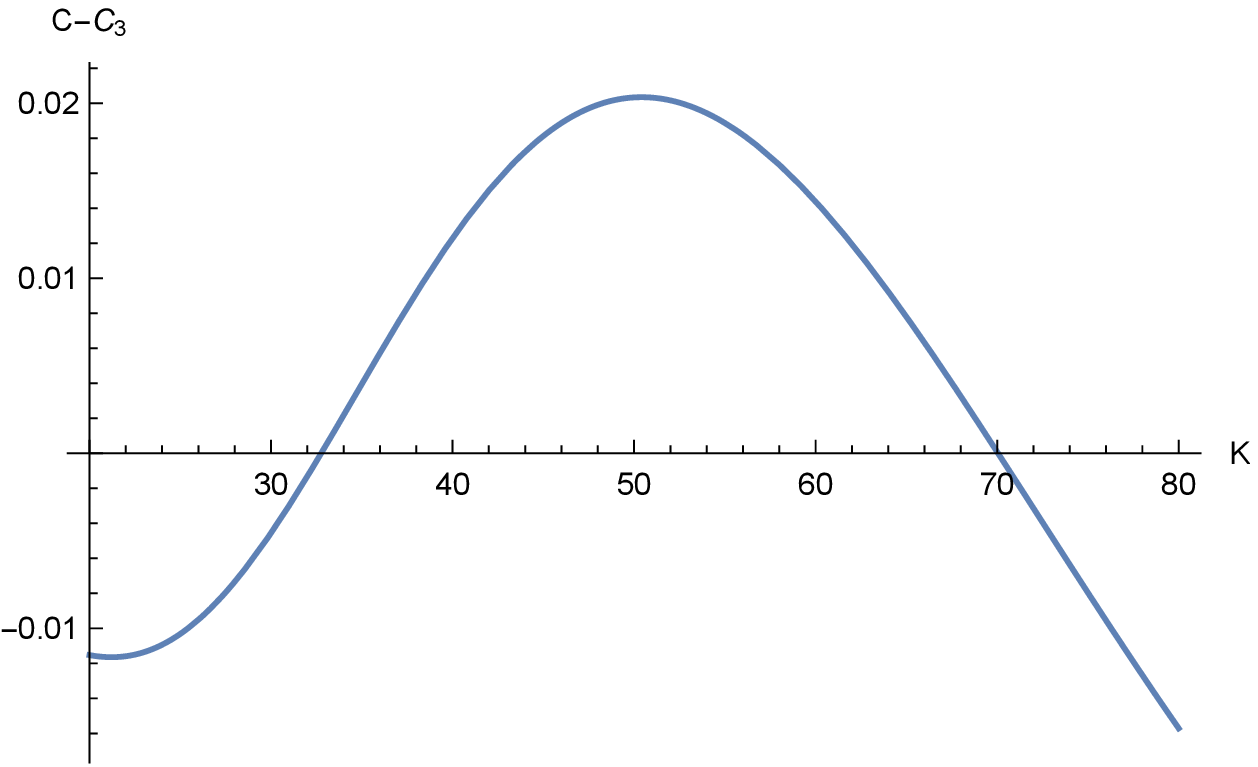}
  \end{minipage}
  }
  \subfigure[The monotonicity of $C_c$ with respect to $\sigma$]{
  \begin{minipage}{7cm}
  \centering
  \includegraphics[width=1.2\textwidth]{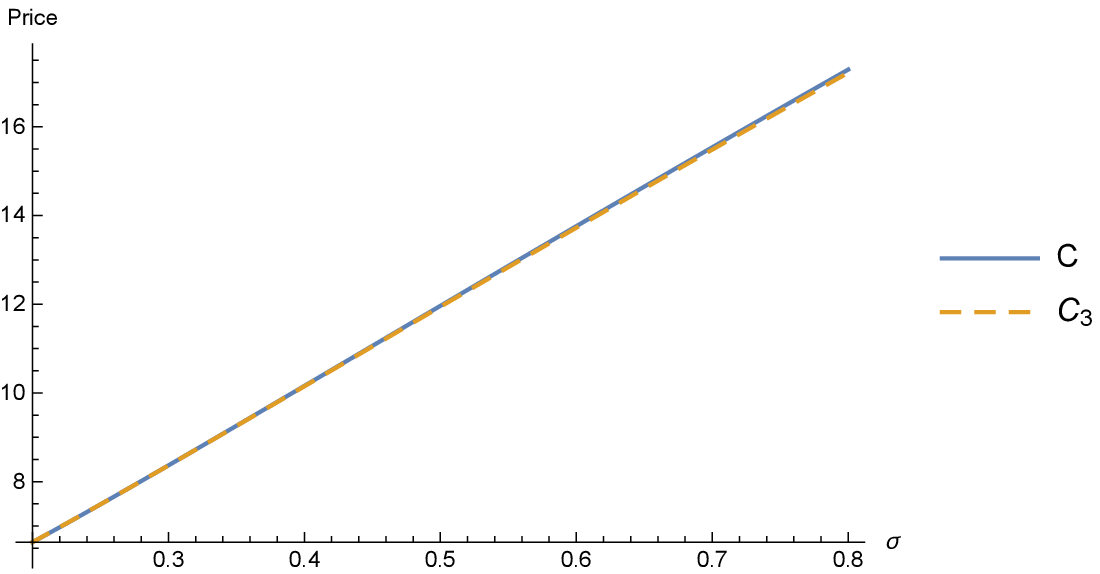}
  \end{minipage}
  }
  \caption{The difference and the monotonicity}
\end{figure}

Now we use the data from Appendix B of \cite{fengler2009arbitrage} to do an empirical test. It is an implied volatility table of the options on DAX index, June 13, 2000. The DAX spot price on that day is $S=7268.91$. In fact, to price options on stock indices or currency options, we need to generalize our new formula by replacing $S$ by $Se^{-qT}$ so that it can deal with options on stocks paying known dividend yields $q$. For further discussion of this, see e.g., \cite[p. 402]{hull2016options}. Fortunately, we do not need to do this here because DAX is a performance index with $q=0$. We still choose $c=3$ and the result is shown in Figure 4.
\begin{figure}[H]
  \centering
  \includegraphics[width=0.8\textwidth]{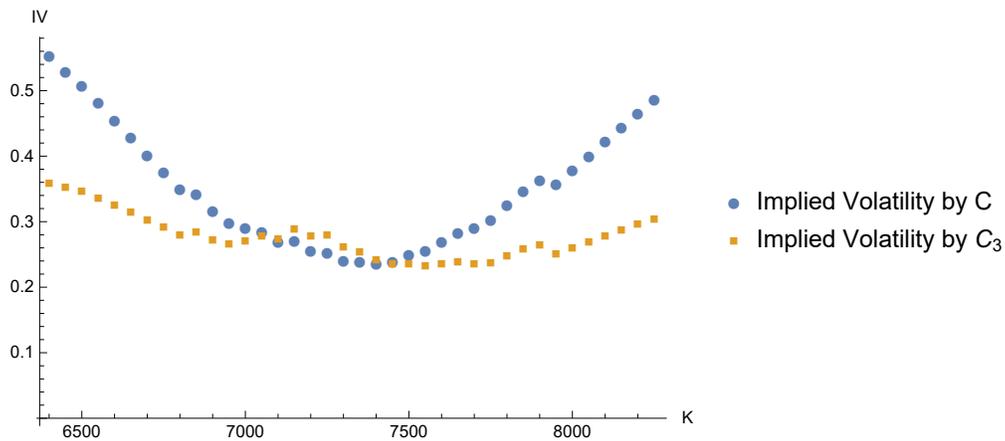}
  \caption{Flatten the smile with $C_3$}
\end{figure}
As expected, our new formula can flatten the smile greatly. With $c=3$, all the points are pushed to around $0.27$. Notice that the points around $S=7268.91$ rise up slightly while the points far from $S=7268.91$ fall down obviously. It is exactly what we want to flatten a standard smile. In fact, the smile here is still a little different from a standard smile since the lowest point (around 7400) is slightly bigger than $S=7268.91$. However, according to
\cite[p. 149]{derman2016volatility}, the smile of those currency options between ``equally powerful" currencies (e.g. USD, EUR) can be a standard one. In that case, our new formula may produce a horizontal line.

The choice $c=3$ in this subsection is quite arbitrary but it has done a good job. Recall that the real velocity of light is about $3\times 10^8$. It is an interesting coincidence. Although $c$ is a constant in physics, we need not to fix $c=3$ from now on. As we know, light travels at different speed in different mediums. At the same time, different markets have different smiles \cite[p. 131]{derman2016volatility}. Therefore, it is reasonable for us to choose different $c$ when we deal with different kinds of options. For example, when we deal with volatility skews of equity options which are more common than smiles, we may need to choose a smaller $c$ in the next subsection.

\subsection{Distributions and Skews}
We use the formula from \cite{breeden1978prices} to obtain a new family of distribution for stock price,
\begin{equation}
g_c(K;r,\sigma,S,T)=e^{rT}\frac{\partial^2 C_c}{\partial K^2}=e^{(r-\beta)T}\frac{1}{2\pi}\frac{S^\alpha}{K^{1+\alpha}}\int_\mathbb{R}(S/K)^{iy}e^{-(\sqrt{c^2y^2+\nu^2}-\nu)T}dy
\end{equation}
where we do differentiation under the integral sign. We have done the same thing when we prove lemma 2.1. So its validity can be verified in a similar way. Recall the lognormal distribution in the Black-Scholes model with the same parameters,
\begin{equation}
g(K;r,\sigma,S,T)=\frac{1}{K\sqrt{2\pi\sigma^2T}}e^{-\frac{(\log(K/S)+\alpha\sigma^2T)^2}{2\sigma^2T}}.
\end{equation}
We set $c=0.6$ in this subsection and compare (17) with (18) in the following Figure 5(a). Figure 5(b) is from \cite[p. 465]{hull2016options} which is used to explain volatility skew of equity options by comparing implied distribution with lognormal distribution.

\begin{figure}[H]
\centering
  \subfigure[$g_c$ and lognormal]{
  \begin{minipage}{7cm}
  \centering
  \includegraphics[width=1.1\textwidth]{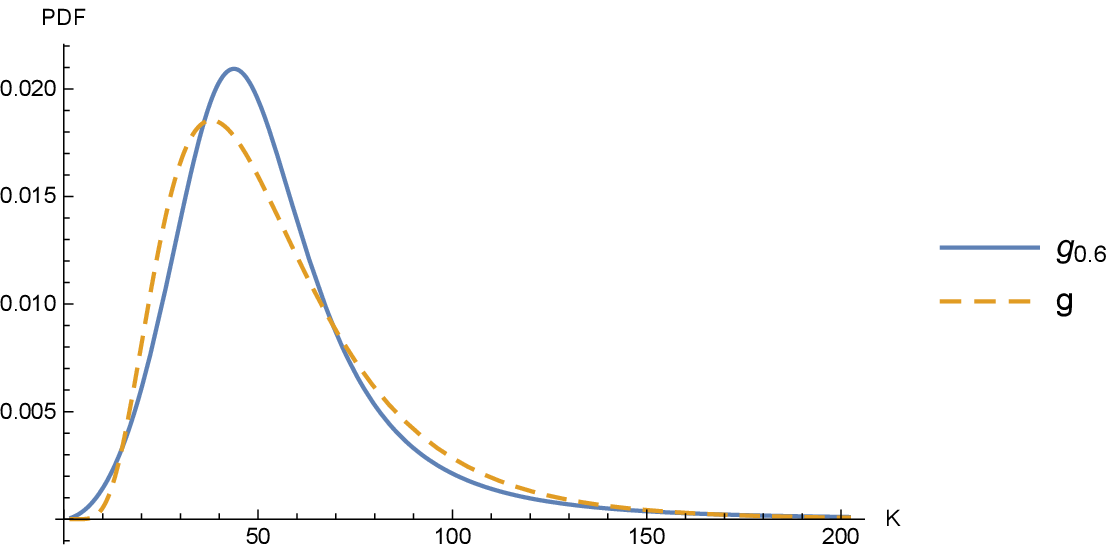}
  \end{minipage}
  }
  \subfigure[implied distribution and lognormal]{
  \begin{minipage}{7cm}
  \centering
  \includegraphics[width=0.8\textwidth]{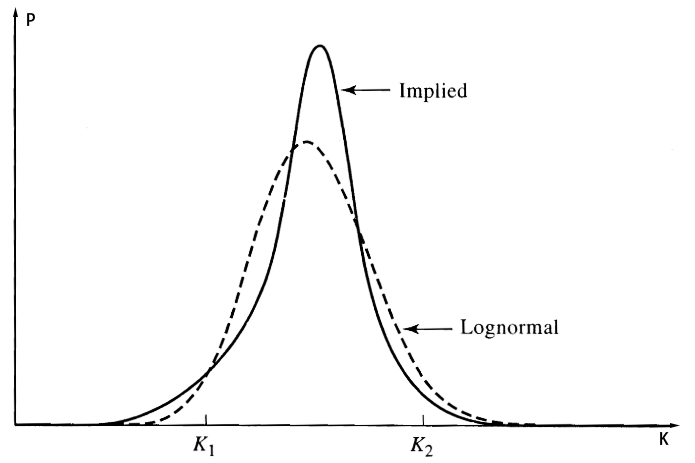}
  \end{minipage}
  }
  \caption{$g_c$, implied distribution and lognormal}
\end{figure}

To our surprise, when comparing with lognormal, $g_c$ in Figure 5(a) has exactly the same features as the implied distribution in Figure 5(b) such as a heavier left tail near $K_1$, a sharper peak with deviation to the right, and a less heavy tail near $K_2$. It is interesting to point out that $g_c$ a heavier right tail than lognormal when $K$ is extremely larger than $K_2$. However, for implied distribution, it is impossible to observe this phenomenon since nobody will trade those options with ridiculously high strike prices. That may be a limitation of option-implied distributions of stock price.

For the convenience of readers, we are going to present the argument in \cite[p. 465]{hull2016options} briefly here to show how Figure 5(b) as well as Figure 5(a) can explain downward volatility skew of equity opitons.

Consider a call with a high strike price $K_2$. Only when the stock price go over $K_2$ can the option pay off. Such probability is lower for the implied distribution than for lognormal which means the Black-Scholes formula overprices those options. By the monotonicity, the implied volatility is lower than $\sigma$ in the Black-Scholes model.

Consider a put with a low strike price $K_1$. Only when the stock price drop below $K_1$ can the option pay off. Such probability is higher for the implied distribution than for lognormal which means the Black-Scholes formula underprices those options. By the monotonicity, the implied volatility is higher than $\sigma$ in the Black-Scholes model.

Therefore Figure 5(b) is consistent with downward volatility skew so is Figure 5(a).

Then we use the data from \cite[p. 97]{daroczi2013introduction} to do an empirical test. It's a set of prices of Google call options, June 25, 2013. The price of Google stock on that day is $S=866.2$. We use smooth curves to fit the volatility skew. Using the Black-Scholes formula, we obtain a smooth curve of option prices to which we can apply the formula from \cite{breeden1978prices} again to calculate the implied distribution. In fact, the volatility curve here still looks like a smile but its lowest point is at around $1050$ which is much larger than $S=866.2$ so we can call it a skew instead. The result is shown in the following figure.
\begin{figure}[H]
  \centering
  \includegraphics[width=0.6\textwidth]{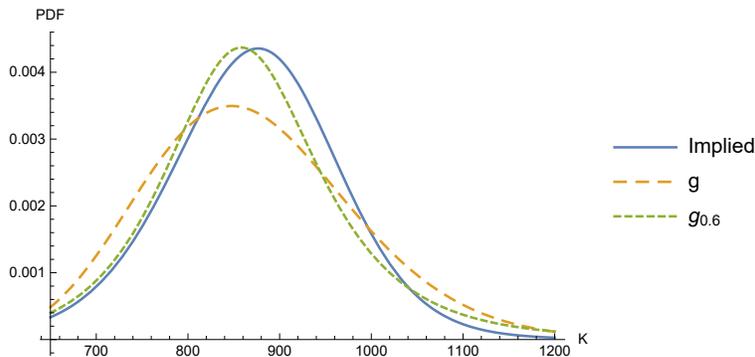}
  \caption{$g_c$, implied distribution of Google stock and lognormal}
\end{figure}
\noindent It is quite obvious that $g_c$ is much closer to the implied distribution than lognormal. Then we can see the advantage of the new distribution family $\{g_c|c\geq c_0\}$. It has just one more parameter $c$ than lognormal while it has almost all the desired features (see Figure 5) to fit those implied distributions from markets. To the best of the author's knowledge, there was no distribution family with such properties before. Finally, we prove the following convergence result so that $\{g_c|c\geq c_0\}$ can be regarded as a powerful generalization of lognormal.
\begin{proposition}
$g_c\rightarrow g$ as $c\rightarrow\infty$.
\end{proposition}
\begin{proof}
Denote the integrand in (17) with $f(y,c)$ (including the constant before the integral). Repeat the argument in the proof of Proposition 1, then we can take the limit into the integral sign. Then we have
\begin{equation}
\begin{aligned}
\lim_{c\rightarrow\infty}g_c&=\lim_{c\rightarrow\infty}\int_\mathbb{R}f(y,c)dy=\int_\mathbb{R}\lim_{c\rightarrow\infty}f(y,c)dy\\
&=e^{-\frac{\alpha^2\sigma^2T}{2}}\frac{(S/K)^\alpha}{K\cdot2\pi}\int_\mathbb{R}e^{-i\log{(K/S)}y}e^{-\frac{y^2\sigma^2T}{2}}dy\\
&=e^{-\frac{\alpha^2\sigma^2T}{2}}\frac{(S/K)^\alpha}{K\cdot\sqrt{2\pi}}\frac{1}{\sqrt{\sigma^2T}}e^{-\frac{\log{(K/S)}^2}{2\sigma^2T}}\\
&=\frac{1}{K\sqrt{2\pi\sigma^2T}}e^{-\frac{(\log(K/S)+\alpha\sigma^2T)^2}{2\sigma^2T}}=g
\end{aligned}
\end{equation}
where we have used the well-known result $(a>0)$ $$\widehat{e^{-\frac{t^2}{2a^2}}}=a\cdot e^{-\frac{\omega^2a^2}{2}}.$$
\end{proof}
\section{Summary}
The relativistic generalized Black-Scholes equation is solved and its closed-form solutions are obtained. New closed-form option pricing formulas are established, containing a new parameter $c$ which can be interpreted as the speed limit of information transfer in markets. It is a generalization of Black-Scholes Formulas in the sense of pointwise convergence. Using new formulas to imply volatility, standard volatility smiles can be flattened. A generalization of lognormal distribution arises from new formulas, also containing $c$. With only one more parameter than lognormal, new distributions have almost all desired features to fit with option-implied distributions of stock price. In this way, downward volatility skews can also be explained.

\clearpage

\bibliographystyle{jfe}
\bibliography{bibfile}

\clearpage

\end{document}